\newif\ifshort
\shortfalse
\documentclass[11pt]{article}
\usepackage[a4paper, total={16cm, 24cm}]{geometry}
\date{}
\usepackage{natbib}
\usepackage{hyperref}
\usepackage{xcolor}
\hypersetup{
	colorlinks=true,       
    linkcolor=blue!70!black,        
    citecolor=blue!70!black,         
    urlcolor=black
}
\renewcommand\tableofcontents{\listoftoc*{toc}} 

\usepackage{authblk}

\author[1]{Piotr Faliszewski}
\author[1]{Stanisław Kaźmierowski}
\author[2]{Grzegorz Lisowski} 
\author[3,4]{Ildikó~Schlotter}
\author[5]{Paolo Turrini}

\affil[1]{AGH University of Kraków, Poland} 
\affil[2]{University of Groningen, The Netherlands} 
\affil[3]{ELTE Centre for Economic and Regional Studies, Hungary} 
\affil[4]{Budapest University of Technology and Economics, Hungary} 
\affil[5]{University of Warwick, United Kingdom} 

\usepackage{times}  
\usepackage{helvet}  
\usepackage{courier}  
\usepackage{natbib}  
\usepackage{caption} 
\usepackage{algorithm}
\usepackage{algorithmic}
\usepackage{booktabs}
\usepackage{newfloat}
\usepackage{listings}
\DeclareCaptionStyle{ruled}{labelfont=normalfont,labelsep=colon,strut=off} 
\floatstyle{ruled}
\newfloat{listing}{tb}{lst}{}
\floatname{listing}{Listing}
\title{Computing Equilibrium Nominations in Presidential Elections}

\usepackage{graphicx} 
\usepackage{caption}
\usepackage{amsmath}
\usepackage{amsthm}
\usepackage{tikz}
\usepackage{amssymb}
\usepackage{thm-restate}
\usepackage[title]{appendix}
\usepackage{MnSymbol}

\newtheorem{lemma}{Lemma}
\newtheorem{definition}{Definition}
\newtheorem{observation}{Observation}
\newtheorem{theorem}{Theorem}
\newtheorem{corollary}{Corollary}
\theoremstyle{definition} 
\newtheorem{remark}{Remark}
\newtheorem{example}{Example}
\AtBeginEnvironment{example}{%
  \pushQED{\qed}%
}
\AtEndEnvironment{example}{\popQED\endexample}

\AtBeginEnvironment{remark}{%
  \pushQED{\qed}%
}
\AtEndEnvironment{remark}{\popQED\endexample}

\newcommand{\linkproof}[1]{%
    \hyperref[#1]{$\star$}%
}

\renewcommand{\phi}{\varphi}
\renewcommand{\leq}{\leqslant}
\renewcommand{\geq}{\geqslant}

\usepackage{listofitems}
\definecolor{darkred}{rgb}{0.64,0,0}
\definecolor{darkcyan}{rgb}{0,0.55,0.55}
\newcommand{\rowcolor}[1]{\textcolor{black}{#1}}
\newcommand{\columncolor}[1]{\textcolor{black}{#1}}

\newcommand{\nfgame}[1]{%
\setsepchar{ }
\readlist\arg{#1}
\begin{tikzpicture}[scale=0.5]
	\node (RT) at (-2,1) [label=left:\rowcolor{\arg[1]}] {};
\node (RB) at (-2,-1) [label=left:\rowcolor{\arg[2]}] {};
\node (CL) at (-1,2) [label=above:\columncolor{\arg[3]}] {};
\node (CR) at (1,2) [label=above:\columncolor{\arg[4]}] {};
\node (RTL) at (-1.4,0.6) {\rowcolor{\arg[5]}}; 
\node (CTL) at (-0.6,1.4) {\columncolor{\arg[6]}}; 
\node (RBL) at (-1.4,-1.4) {\rowcolor{\arg[7]}};
\node (CBL) at (-0.6,-0.6) {\columncolor{\arg[8]}};
\node (RTR) at (0.6,0.6) {\rowcolor{\arg[9]}};
\node (CTR) at (1.4,1.4) {\columncolor{\arg[10]}};
\node (RBR) at (0.6,-1.4) {\rowcolor{\arg[11]}};
\node (CBR) at (1.4,-0.6) {\columncolor{\arg[12]}};
\draw[-,very thick] (-2,-2) to (2,-2);
\draw[-,very thick] (-2,0) to (2,0);
\draw[-,very thick] (-2,2) to (2,2);
\draw[-,very thick] (-2,-2) to (-2,2);
\draw[-,very thick] (0,-2) to (0,2);
\draw[-,very thick] (2,-2) to (2,2);
\draw[-,very thin] (-2,2) to (0,0);
\draw[-,very thin] (0,0) to (2,-2);
\draw[-,very thin] (-2,0) to (0,-2);
\draw[-,very thin] (0,2) to (2,0);
\end{tikzpicture}}

\def\pospres{\textsc{Possible President}}
\def\necpres{\textsc{Necessary President}}
\def\findNashEq{\textsc{Equilibrium President}}
\def\Nashexistence{\textsc{Equilibrium Existence}}
\def\scheme{scheme}
\def\schemes{schemes}
\def\P{\mathcal{P}}
\def\E{\mathcal{E}}
\def\S{\mathcal{S}}

\def\NP{\mathsf{NP}}
\def\coNP{\mathsf{coNP}}

\usepackage{pict2e}
\newcommand{\opentriangle}{%
  \raisebox{0.2pt}{\makebox[0.77778em]{%
    \setlength{\unitlength}{0.6em}%
    \linethickness{0.4pt}\roundjoin
    \begin{picture}(1,1)
    \polygon(0,0)(1,0)(1,1)
    \end{picture}%
  }}%
}
\newenvironment{proofsketch}{\par
  \pushQED{\hfill \opentriangle}%
   \noindent 
   {\itshape Proof sketch.} \ignorespaces
}{%
  \popQED\endtrivlist
}

\usepackage{cleveref}
\DeclareRobustCommand{\abbrevcrefs}{%
  \Crefname{theorem}{Thm.}{Thms.}%
  \Crefname{example}{Ex.}{Exs.}%
  \Crefname{proposition}{Pr.}{Prs.}%
  \Crefname{corollary}{Cor.}{Cors.}%
}

\DeclareRobustCommand{\Shcref}[1]{{\abbrevcrefs\Cref{#1}}}

\usepackage{arydshln}

\begin{document}

\maketitle

\begin{abstract}
We study strategic candidate nomination by parties in elections decided by Plurality voting. Each party selects a nominee before the election, and the winner is chosen from the nominated candidates based on the voters' preferences. We introduce a new restriction on these preferences, which we call \textit{party-aligned single-peakedness}: all voters agree on a common ordering of the parties along an ideological axis, but may differ in their perceptions of the positions of individual candidates within each party. The preferences of each voter are single-peaked with respect to their own axis over the candidates, which is consistent with the global ordering of the parties.
We present a polynomial-time algorithm for recognizing whether a preference profile satisfies party-aligned single-peakedness. 
In this domain, we give polynomial-time algorithms for deciding whether a given party can become the winner under some (or all) nominations, and whether this can occur in some pure Nash equilibrium. 
We also prove a tight result about the guaranteed existence of pure strategy Nash equilibria for elections with up to three parties for single-peaked and party-aligned single-peaked preference profiles.
\end{abstract}



\section{Introduction}

Let us consider a certain computer science department that needs to
select its new head. The department has three research groups, one
focused on artificial intelligence, one working on theoretical computer
science, and one devoted to distributed systems. In the past, each
group had just one candidate, and the person who got the most votes was
selected as the head, for a four-year term. However, this time more people
expressed their interest in the position; the groups---worried by the
possibility of splitting the vote---decided that each of them will nominate
only a single one.
%
Yet, how should they decide on whom to choose?

\citet{faliszewski2016} studied this problem by looking for a
\emph{necessary winner}, i.e., a candidate who wins irrespective of
the nominations by the other groups, or by looking for a
\emph{possible winner}, i.e., a candidate who wins, provided the other
groups' nominations are favourable (formally, they referred to such
candidates as \emph{necessary} and \emph{possible president}, as they
followed a naming scheme from politics). Yet, asking for a necessary
winner is too demanding---such a candidate may fail to exist---and
choosing a possible winner is too optimistic---there is no reason for
the other groups to cooperate. Hence, we are looking for a
game-theoretic solution: we model the scenario as a game, where each
group is a player, each possible nominee is a strategy, and we are
looking for a pure strategy Nash equilibrium (a group's utility is $1$ if its
nominee wins and it is $0$ otherwise), focusing on the classic
Plurality rule (whoever gets the most votes wins). We are interested in
two types of results: we analyse when equilibria are guaranteed to
exist and, if that is not the case, what is the complexity of deciding
their existence. Henceforth, we follow the politics-based
naming convention of \citet{faliszewski2016}: we speak of parties that
nominate candidates, and we refer to the candidates who win under a
given pure strategy Nash equilibrium as \emph{equilibrium presidents}.

\paragraph{Our Contribution.} We find that the guarantees for
equilibria existence, as well as the complexity of recognising
equilibria, are very fragile and strongly depend on the nature of the
voters' preferences. In particular, we distinguish three types of
preference profiles that the voters may have. First, we analyse
1D-Euclidean preferences~\citep{ene-hin:b:spatial,ene-hin:b:spatial2},
where each candidate and each voter is a point on the line, and the
voters rank the candidates in the order of their distance from their
points. Second, we consider the classic single-peaked
setting~\citep{bla:b:polsci:committees-elections}, where there is a
commonly agreed order (or axis) of the candidates---such as the
political left-to-right spectrum---and every prefix of every vote
forms an interval within this axis. Finally, we introduce the notion
of {\em party-aligned single-peakedness}.


In party-aligned single-peaked preferences, parties are ordered along
a given axis, but each voter has their perceived axis, obtained
from the party axis by replacing each party with its members, in
whatever order the voter prefers.  In other words, all voters agree on
the positioning of the parties along the party axis, but they may
disagree on the precise positioning of the parties' candidates. Such
differences may stem, for example, from differing beliefs about some
better known or local candidate, whom a voter perceives as more (or
less) aligned with their ideals than other candidates from the same
party. For example, consider parties $A=\{a_1,a_2\}$ and $B=\{b\}$ and
the next three votes:
\[
    v_1 : a_2 \succ b \succ a_1; \quad 
    v_2 : a_1 \succ b \succ a_2;
    \quad
    v_3 : a_1 \succ a_2 \succ b
\]
based on their following perceived axes:
\[
    v_1,v_3 : a_1 \triangleleft a_2 \triangleleft b; 
    \qquad
    v_2 : a_2 \triangleleft a_1 \triangleleft b.
\]
These preferences are not single-peaked---as they rank three
different candidates in the last position, which is impossible in
single-peaked elections---but they are party-aligned
single-peaked. The mismatch in voters' perception can give rise to
non-trivial strategic decisions from the parties' side.  While
party~$B$ can only nominate~$b$, party~$A$ has to make a choice
between~$a_1$ and~$a_2$.
If party~$A$ nominates~$a_2$, voter~$v_2$ perceives party $B$ as
preferable to party~$A$. 
By contrast, this is not the case
for~$v_1$, who prefers~$a_2$ to~$b$. However, should party $A$
nominate $a_1$ instead, they would end up gaining the vote of $v_2$
but losing that of $v_1$ to party $B$.
Essentially, under party-aligned single-peakedness, a move to steal
voters from one party can lead to losing some other ones.

We defer some of our proofs to
\ifshort the 
full version of the paper~\cite{fullversion}, 
\else
the appendix 
\fi
and mark such  results with the~$\star$ symbol.

\begin{remark}
  Though neither 1D-Euclidean nor single-peaked preferences need to
  put members of a given party on consecutive positions on the common
  axis (or line), all voters have to follow the same one. By contrast, party-aligned single-peaked preferences require members
  of each party to take consecutive positions on each axis, but
  voters can freely order specific party members along their
  individual axes.
\end{remark}

We establish the following results:
\begin{enumerate}
\item Pure strategy Nash equilibria always exist, provided there are at most three parties and the preferences are both single-peaked and
  party-aligned single-peaked, but may fail to exist for 
  four parties (even with at most two members in each);
  for 1D-Euclidean elections, equilibria may fail to exist even for
  two parties of size at most two.

\item Whether a pure strategy Nash
  equilibrium exists in a party-aligned single-peaked election can be decided in polynomial time, but
  the problem 
  is $\NP$-hard in 1D-Euclidean elections (and, hence,
  under single-peaked ones). 
  Deciding if a given party
  can win in some pure strategy Nash equilibrium, or if it can win under some or all
  nominations from other parties,
  is also polynomial-time solvable in party-aligned single-peaked elections.
\end{enumerate}

\begin{remark}
  We focus on pure strategy Nash equilibria, which we henceforth simply refer to as Nash equilibria.
  Although allowing for randomisation is an important aspect of strategic decisions, we believe ``pure" candidate selections are natural and often occurring events in the decision-making process of parties, e.g., through the use of primaries \cite{borodin2024primarily}.
Besides related models of strategic nominations \cite{HarrensteinLST21, sabato2017real}, the focus on pure strategy Nash equilibria is also widely established in other areas of computer science, such as the logical analysis of strategic reasoning in games, including, e.g., Boolean games~\cite{AgotnesHHW13, GutierrezHPW21}. 
\end{remark}

\paragraph{Related Literature.}

The line of research on party-based election models, as introduced by \citet{faliszewski2016}, has been expanded by taking into account different voting rules and investigating the related problems also from the parameterised complexity viewpoint by \citet{misra2019parameterized}, \citet{cechlarova2023hardness}, \citet{schlotter2024parameterized}, and \citet{ildi_aamas2025}.
In a related vein, \citet{pie-szu:c:alliances} analysed a family of rules where nominations are not needed.

 Finding a Nash equilibrium was also studied in a similar model describing strategic nominations over different territories by \citet{HarrensteinT22}. 
%
Equilibrium existence and computation for nominee selection problems was further studied 
in the context of the Hotelling-Downs framework where both voters and candidates are located on a line, which induces a natural preference ordering of candidates \cite{HarrensteinLST21,deligkas2022parameterized}. 
In a related paper, \citeauthor{sabato2017real} restricts candidates' and voters' positioning to intervals, which has strong implications on equilibrium existence.
Decision-making models in which strategic nomination has been studied include tournaments played by coalitions of players \cite{Lisowski0T22,lisowski2022strategic}.


Another important related line of research is \emph{strategic candidacy}; see \cite{brill2015strategic, dutta2001strategic, eraslan2004strategic}. There, candidates have preferences over their opponents and may decide not to participate in an election to achieve a more favoured outcome. Such strategic behaviour may arise, e.g., due to certain costs incurred by participation \cite{obraztsova2015strategic,elkind2015equilibria}. Finally, we mention that our research is directly related to the study of nominee selection in primaries; see, e.g., the work of \citet{borodin2019primarily}.

Structured domains are surveyed in depth by
\citet{elk-lac-pet:t:restricted-domains-survey}. A variant of
single-peakedness closely related to ours  is the domain due to
\citet{cor-gal-spa:c:sp-width,cor-gal-spa:c:spsc-width} where each voter has to
rank members of each party consecutively; note that this is not required for party-aligned single-peakedness.

\section{Preliminaries}

\subsection{Strategic Nominee Selection}

For an integer~$k\geq 0$ let $[k]=\{1,\dots,k\}$ with $[0]=\emptyset$ .

Consider a set~$C$ of candidates and a set $\mathcal{P} = \{P_1, \ldots, P_k\}$ of  \emph{parties} where $\P$ is a partitioning of the candidate set~$C$.
%
Let $V = \{v_1, \ldots, v_n\}$ be a set of voters, where each voter~$v$  has a strict preference order~$\succ_v$ over the candidate set~$C$; we call 
the collection $(\succ_v)_{v \in V}$ \emph{preference profile} of the voters.

Our framework consists of two phases: a {\em nomination phase}, where each party selects its candidate, and an {\em election phase}, where voters choose among the nominated candidates based on their preferences.
Thus, a \emph{nomination \scheme} is a tuple $(c_1, \ldots, c_k)$ where $c_i \in P_i$ for all $i \in [k]$.
Given a nomination \scheme\ $(c_1, \ldots, c_k)$, the \emph{reduced election} takes place over the set $C' = \{c_1, \ldots, c_k\}$ of 
\emph{nominees}. 
We use the \emph{Plurality} voting rule where each voter \emph{votes for}  their favourite candidate among~$C'$, and the candidates with the highest number of votes obtained---their \emph{score}---are the \emph{winners}; the party whose nominee is a winner 
is also called a \emph{winner}. 
Some candidate or party is a winner in some nomination scheme~$(c_1,\dots,c_k)$ if it is a winner in the reduced election over~$\{c_1,\dots,c_k\}$.

%

 We say that a candidate~$c'_i \in P_i$ is a \emph{Nash deviation} by party~$P_i$ for some nomination \scheme\ $(c_1,\dots,c_k)$ if $c'_i \in P_i \setminus \{c_i\}$, and $c_i$ is not a winner in the election resulting from $(c_1,\dots,c_k)$
 but $c'_i$ is a winner in the election resulting from $(c_1,\dots,c'_i,\dots,c_k)$.
 A nomination \scheme\ 
 is a \emph{(pure) Nash equilibrium} if it admits no Nash deviation.

\subsection{Structured Preferences}


Let us now introduce two types of restrictions on voters' preference profiles that may arise in elections where there is a single ideological line which strongly determines voters' behaviour such as, e.g., the left-to-right axis in politics.

\paragraph{Single-Peaked Profiles.} 
    We say that a preference profile $(\succ_v)_{v \in V}$ is \emph{single-peaked} if there exists a total order~$\triangleleft$ over the set~$C$ of candidates called an \emph{axis} 
such that the preferences of each voter~$v$ are \emph{single-peaked with respect to the axis $\triangleleft$}, meaning that for every three candidates $a,b,c \in C$ with $a \triangleleft b \triangleleft c$ the relation $a \succ_v b$ implies $b \succ_v c$.
Roughly speaking, this means that candidates can be placed along a line such that each voter~$v$' preferences are derived from the candidates' distances to $v$'s ``ideal point'' on the line.
%
%
 
\paragraph{1-D Euclidean Profiles.}
A special case of single-peaked profiles appears in \emph{1-D Euclidean elections} where voters and candidates are located on the real line, and voters' preferences over candidates stem from their distances. 
\begin{definition}\label{def:Euclidian}
We say that an election $E=(C,V)$ is \emph{1-D Euclidean} if there is a mapping $f: C \cup V \rightarrow \mathbb{R}$, such that for every voter $v$ in $V$ and every pair $c_i, c_j$ in $C$, $c_i \succ_v c_j$ if and only if $|f(v) - f(c_i)| < |f(v) - f(c_j)|$.
\end{definition}
Although 1-D Euclidean elections allow for voters to be indifferent between candidates that are equally close to them, we will ensure that the 1-D Euclidean elections used in this paper induce strict preferences.


\section{Party-Aligned Single-Peakedness}

We assume that parties can be ordered along a global axis recognised by all voters, with candidates belonging to the same party placed contiguously. 
However, candidates within a party may be perceived differently by voters, who therefore may disagree on where different candidates within a given party are located along the axis. 
Formally, we say that a preference profile is \emph{party-aligned single-peaked} if there exists a \emph{party axis}~$\triangleleft_{\mathcal{P}}$, i.e., a total order over the set~$\mathcal{P}$ of parties
%
%
such that, for every voter $v \in V$, there exists a \emph{perceived axis} $\triangleleft_v$ over the candidate set~$C$ 
for which 

\begin{itemize}
  \item $v$'s preference order~$\succ_v$ is single-peaked w.r.t.~$\triangleleft_v$, and
  \item $\triangleleft_v$ is \emph{compatible} with $\triangleleft_\P$, meaning that 
for each two candidates~$c_i$ and~$c_j$ belonging to different parties~$P_i$ and~$P_j$, respectively, $c_i \triangleleft_v c_j$ implies $P_i \triangleleft_{\mathcal{P}} P_j$.
\end{itemize}
%

\begin{example}
Consider a set~$\P=\{P_a,P_b,P_c,P_d\}$  of parties with $P_a=\{a\}$, $P_b=\{b_1,b_2\}$,
$P_c=\{c_1,c_2\}$, and $P_d=\{d\}$, and three voters whose preferences are:
\begin{align*}
    v &: c_2 \succ b_1 \succ c_1 \succ d \succ b_2 \succ a; \\
    w &: b_1 \succ c_1 \succ b_2 \succ c_2 \succ a \succ d; \\
    z &: b_2 \succ c_2 \succ c_1 \succ d \succ b_1 \succ a. 
\end{align*}
The voters' preferences are party-aligned single-peaked with party axis $P_a \triangleleft_\P P_b  \triangleleft_\P P_c  \triangleleft_\P P_d$ and perceived axes:

\begin{center}$
\quad\quad\quad\quad \begin{array}[b]{l@{\hspace{1pt}}
l@{\hspace{1pt}}l@{\hspace{1pt}}l@{\hspace{1pt}}l@{\hspace{1pt}}l}
    a \,\,\triangleleft_{v} & b_2 \,\,\triangleleft_{v} & b_1 \,\,\triangleleft_{v} & c_2 \,\,\triangleleft_{v} & c_1 \,\,\triangleleft_{v} & d; \\
    a \,\,\triangleleft_{w} & b_2 \,\,\triangleleft_{w} & b_1 \triangleleft_{w} & c_1 \,\,\triangleleft_{w} & c_2 \,\,\triangleleft_{w} & d; \\
    a \,\,\triangleleft_{z} & b_1 \,\,\triangleleft_{z} & b_2 \,\,\triangleleft_{z} & c_2 \,\,\triangleleft_{z} & c_1 \,\,\triangleleft_{z} & d. 
\end{array}
\quad\quad\quad\phantom{.}
 \mbox{\qedhere}$
\end{center}
\end{example}

   \subsection{Recognising Party-Aligned Single-Peaked Profiles.}
\label{sec:recognition}
We present a polynomial-time algorithm for recognising preference profiles that are party-aligned single-peaked. Formally, we deal with the following computational problem.

\medskip
\noindent
\begin{minipage}{0.88\columnwidth}
\fbox{
    \parbox{1.08\columnwidth}{
    \textsc{Recognising party-aligned single-peakedness$\!\!\!\!\!\!\!\!$} 
    \begin{description}
        \item[{\bf Input:}]
            A set~$C$ of candidates partitioned into a set~$\mathcal{P}$
            of parties, and a set~$V$ of voters with preference profile $(\succ_v)_{v \in V}$ over the candidate set~$C$. 
        \item[{\bf Question:}]
            Is the preference profile $(\succ_v)_{v \in V}$ party-aligned single-peaked?
    \end{description}
    }
    }
\end{minipage}
\medskip

The rest of this section proves the following result.
\begin{theorem}
    \label{thm:recognition-poly} 
    \textsc{Recognizing party-aligned single-peakedness} 
    can be solved, and a suitable party axis---if existent---can be computed, in $O(|C| \cdot |V| \cdot |\mathcal{P}|)$ time.
\end{theorem}

\subsubsection{Verification of Party-Aligned Single-Peakedness Under Fixed Party Ordering.}

First, we formulate the necessary and sufficient conditions for a single-vote profile to be party-aligned single-peaked under a fixed ordering 
of parties. 

\begin{restatable}[\linkproof{sec:proof-of-lemvotePSP}]{lemma}{lemvotepartysinglepeaked}
\label{lem:vote_party_single_peaked}
Consider a vote~$\succ_v$ over the candidate set of parties $P_1,\dots,P_m$.
For each party~$P_i$, let $h_i$ and~$l_i$ denote $v$'s most- and least-preferred candidate within~$P_i$. 
Let $P_j$ be the party containing $v$'s top candidate.
Then $\succ_v$ is party-aligned single-peaked with party axis $P_1 \triangleleft_\P \dots \triangleleft_\P P_m$ if and only if the following conditions are met.
\begin{enumerate}
    \item[(a)] If $P_j \notin \{P_1,P_m\}$, then $l_j \succ_v h_{j-1}$ or $l_j \succ_v h_{j+1}$.
    \item[(b)] For each parties $P_i, P_{i+1}$ with $i>j$ we have $l_i \succ_v h_{i+1}$. 
    \item[(c)] For each parties $P_{i-1}, P_i$ with $i < j$ we have $l_i \succ_v h_{i-1}$. 
\end{enumerate}  
\end{restatable}

Notice that whenever there are at most two parties, then conditions (1)--(3) of Lemma~\ref{lem:vote_party_single_peaked} are automatically true for any fixed party axis. Thus, we have the following consequence.
\begin{corollary}
\label{cor:twoparties-always-partySP}
    Every preference profile with at most two parties is party-aligned single-peaked.
\end{corollary}

\subsubsection{Placing the Bottom Candidates.}

Following an idea by~\citet{Escoffier_08}, 
we present a lemma about the placement of a party that contains the lowest-ranked candidate in some vote.

\begin{restatable}[\linkproof{sec:proof-oflembottomparties}]{lemma}{lembottomparties}   
\label{lem:bottom_parties}
     Let $c \in P_i$ be the candidate ranked in the last position by a voter~$v$. If $\succ_v$ is party-aligned single-peaked with  party axis $\triangleleft_{\mathcal{P}}$, then $P_i$ is either in the leftmost or the rightmost position in~$\triangleleft_{\mathcal{P}}$.
\end{restatable}

As a consequence of Lemma~\ref{lem:bottom_parties}, in any party-aligned single-peaked consistent profile, at most two parties may have candidates that are ranked last by at least one voter. 

\subsubsection{Vote-Imposed Restrictions.}
We move on to present some conditions that allow us to extend a partially determined party axis based on the given preference profile. 

Suppose that we already know the position of certain parties in the party axis~$\triangleleft_\P$ to be constructed. 
In particular, assume that we have an \emph{extremal party placement $(L,R)$ for~$\triangleleft_\P$} which consists of two sets of parties $L, R \subseteq \P$ such that $L$ contains the $|L|$ leftmost parties and $R$ contains the $|R|$ rightmost parties according to~$\triangleleft_\P$.


\begin{restatable}[\linkproof{sec:proof-of-partyfixedplacement}]{lemma}{lempartyfixedplacement}
\label{lem:party_fixed_placement}
Let $(L,R)$ be an extremal party placement for some party axis $\triangleleft_\P$, and consider a vote~$\succ_v$ that is party-aligned single-peaked with axis~$\triangleleft_\P$.
Assume that candidates in~$\bigcup (R \cup L)$ do not form the suffix of the vote $\succ_v$.
Let 
\begin{itemize}
    \item $a$ be $v$'s least favourite candidate in $\bigcup{(\mathcal{P} \setminus (L \cup R))}$,
    \item $b$ be $v$'s most favourite candidate in $\bigcup (L \cup R)$, and
    \item $P_a$ and $P_b$ be the parties containing~$a$ and~$b$, respectively.
\end{itemize}
If $P_b \in R$, then $P_a$ directly follows the last party in~$L$ according to~$\triangleleft_\P$, while
if $P_b \in L$, then $P_a$ directly precedes the first party in~$R$ according to~$\triangleleft_\P$.
\end{restatable}

\subsubsection{Reducing the Problem.}
When Lemma~\ref{lem:party_fixed_placement} can no longer be applied, we use recursion based on the following lemma.
\begin{restatable}[\linkproof{sec:proof-of-reducingproblem}]{lemma}{lemreducingproblem}
\label{lem:reducing_problem}
Let $\Pi=(\succ_v)_{v \in V}$ be a party-aligned single-peaked preference profile over candidate set~$C$ with a party axis~$P_1 \triangleleft_P \dots \triangleleft_\P P_m$.
Assume further that 
the candidates in ($P_1 \cup \dots \cup P_i) \cup (P_j \cup \dots \cup P_m)$ form a suffix of each vote in~$\Pi$, and let $Q$ be the set of candidates not in this suffix.
If the restriction of $\Pi$ to $Q$ is party-aligned single-peaked with a party axis ${Q_1 \triangleleft_Q \dots \triangleleft_Q Q_k}$, then $\Pi$ is party-aligned single-peaked with the party axis $P_1 , \dots, P_i , Q_1, \dots,  Q_k , P_j, \dots, P_m $.
\end{restatable}

\subsubsection{Algorithm.}
Armed with the aforementioned results, we present the recognition algorithm proving Theorem~\ref{thm:recognition-poly}. 

\textit{Placing Bottom Parties.}
Given a preference profile $\Pi$, we begin by identifying the set $B \subseteq \mathcal{P}$ of bottom parties, i.e., those containing a candidate appearing at the last position in some vote.  Lemma~\ref{lem:bottom_parties} shows that all parties in~$B$ must appear at either the left- or the rightmost position of the hypothetical party axis~$\triangleleft_\P$. 
In particular, if $|B| > 2$, then the profile is not party-aligned single-peaked, and we reject. 
If $|B| \leq 2$, then the exact placement of the parties in~$B$ is not important, because for every party axis witnessing the party-aligned single-peakedness of our profile, the reversed party axis does the same. Hence, we 
place the parties in~$B$ as the left- and (if $|B|=2$) rightmost parties in~$\triangleleft_\P$, creating an extremal party placement $(L,R)$ for~$\triangleleft_\P$ with~$L \cup R =B \neq \emptyset$.

\textit{Vote-Imposed Restrictions.}
Next, we iterate over the set of voters to check whether they impose some restriction on~$\triangleleft_\P$. More precisely, we apply Lemma~\ref{lem:party_fixed_placement} as long as possible, extending our extremal party placement~$(L,R)$---together with the ordering of the parties in~$L \cup R$ according to~$\triangleleft_\P$---by fixing the placement of an additional party in the fashion determined by Lemma~\ref{lem:party_fixed_placement}  in each step. 


\textit{Reducing the Problem.}
If at some point there is no voter for which  Lemma~\ref{lem:party_fixed_placement} can be applied, then the set of candidates contained in our current extremal party placement $(L,R)$ forms a suffix in each vote. Hence, in line with Lemma~\ref{lem:reducing_problem}, we delete all parties in~$L \cup R$ together with all of their candidates from the profile~$\Pi$, and find a suitable party axis for the resulting profile~$\Pi'$ (containing at least one party less than~$\Pi$) in a recursive manner. Note that if $\Pi$ is party-aligned single-peaked, then so is $\Pi'$, and a party axis witnessing the latter can be turned into a party axis for~$\Pi$ using Lemma~\ref{lem:reducing_problem} and the ordering of the parties contained in the extremal party placement~$(L,R)$ determined so far.

\textit{Checking the Resulting Profile.}
After obtaining a party axis~$\triangleleft_\P$, we check if profile~$\Pi$ is party-aligned single-peaked under~$\triangleleft_\P$ by checking the conditions 
in Lemma~\ref{lem:vote_party_single_peaked}. If $\Pi$ is not party-aligned single-peaked under $\triangleleft_\P$, then Lemmas~\ref{lem:bottom_parties}--\ref{lem:reducing_problem} guarantee that $\Pi$ is not party-aligned single-peaked.

\textit{Running Time.} The algorithm runs in $O(|C| \cdot |V| \cdot |\mathcal{P}|)$ time.
First, identifying the bottom parties can be done in $O(|V|)$ time, and such a step is performed at most $|\mathcal{P}|$ times.
Second, checking the instance for vote-imposed restrictions and then applying Lemma~\ref{lem:party_fixed_placement} takes $O(|C| \cdot |V|)$ time; again, such a step is performed at most $|\mathcal{P}|$ times.
Third, checking whether our instance can be  reduced recursively by Lemma~\ref{lem:reducing_problem} takes $O(|C| \cdot |V|)$ time, not counting the time necessary to deal with the reduced instance. This yields an overall running time of $O(|C| \cdot |V| \cdot |\mathcal{P}|)$.

\section{Equilibrium Existence}
\label{sec:equilibrium-existence}
We move on to present the results concerning the existence of Nash equilibria in the considered model.

\begin{theorem}
\label{thm:NE-existence-three-parties}
If voters' preference profiles are single-peaked and also party-aligned single-peaked with at most three parties, then a Nash equilibrium always exists.    
\end{theorem}

\begin{proof}
    Let $(\succ_v)_{v \in V}$ be a preference profile that is both single-peaked and party-aligned single-peaked, i.e., each vote is single-peaked with respect to the same axis~$\triangleleft$ that is compatible with some party axis. 
    We call a nomination scheme~$\vec{c}$ \emph{centrist} if the leftmost party~$P_L$ nominates its rightmost candidate, and the rightmost party~$P_R$ nominates its leftmost candidate in~$\vec{c}$. 
    If a voter~$v$ does not vote for~$P_L$ in a centrist nomination scheme~$\vec{c}$, then its favourite candidate is to the right of every candidate of~$P_L$, and hence 
    there is no Nash deviation for~$\vec{c}$ by~$P_L$ due to the single-peakedness of~$\succ_v$ w.r.t.~$\triangleleft$. 
    A symmetric argument shows that $\vec{c}$ admits no Nash deviation by~$P_R$ either. 

    We now distinguish between two cases. 
    If there is a candidate~$c \notin P_L \cup P_R$ that is a winner in some centrist nomination scheme~$\vec{c}$, then neither the party containing~$c$, nor any of the parties~$P_L$ and~$P_R$ has a Nash deviation for~$\vec{c}$ which is thus a Nash equilibrium.
    By contrast, if in all centrist nomination schemes, no party other than~$P_L$ and~$P_R$ wins, then any such nomination scheme is a Nash equilibrium, since no party has a Nash deviation.
\end{proof}

\paragraph{Non-Existence of Nash Equilibria.}

Let us show an example proving the tightness of Theorem~\ref{thm:NE-existence-three-parties} in the sense that party-aligned single-peakedness alone does not guarantee a Nash equilibrium, even if the number of parties is only two.
\begin{theorem}
    \label{thm:No-NE-partySP-twoparties}
    There is an election~$E=(C,V)$ 
    containing two size-two parties and three voters
    with a party-aligned single-peaked preference profile  
    that admits no Nash equilibria.
\end{theorem}

\begin{proof}
    Consider the election~$E=(C,V)$ with two parties $A=\{a_1,a_2\}$ and $B=\{b_1,b_2\}$ and voters $V=\{v_1,v_2,v_3\}$ whose preferences are 
    \begin{align*}
        v_1 &: a_1 \succ b_1 \succ a_2 \succ b_2; \\
        v_2 &: b_1 \succ a_2 \succ b_2 \succ a_1; \\
        v_3 &: a_2 \succ b_2 \succ a_1 \succ b_1. 
    \end{align*}
    By Corollary~\ref{cor:twoparties-always-partySP}, this profile is party-aligned single-peaked. To see that there is no Nash equilibrium, it suffices to consider each possible nomination \scheme\ and the number of voters supporting the two parties in each of them, as depicted by Figure~\ref{fig:NonEx-twoparties}.
    It is easy to see that each of the four possible nomination \schemes\ admits a Nash deviation.
    \begin{figure}[t]
	\begin{center}
        \begin{tikzpicture}[scale=0.5]
\node (RT) at (-3,1) [label=left:$a_1$] {};
\node (RB) at (-3,-1) [label=left:$a_2$] {};
\node (CL) at (-1.5,2) [label=above:$b_1$] {};
\node (CR) at (1.5,2) [label=above:$b_2$] {};
\node (RTL) at (-2,0.5) {$v_1,v_3$}; 
\node (CTL) at (-0.9,1.4) {$v_2$}; 
\node (RBL) at (-2,-1.4) {$v_3$};
\node (CBL) at (-1,-0.6) {$v_1,v_2$};
\node (RTR) at (0.9,0.6) {$v_1$};
\node (CTR) at (2,1.4) {$v_2,v_3$};
\node (RBRtop) at (0.6,-0.85) {$v_1,$};
\node (RBRbottom) at (1,-1.5) {$v_2,v_3$};
\node (CBR) at (2,-0.6) {--};
\draw[-,very thick] (-3,-2) to (3,-2);
\draw[-,very thick] (-3,0) to (3,0);
\draw[-,very thick] (-3,2) to (3,2);
\draw[-,very thick] (-3,-2) to (-3,2);
\draw[-,very thick] (0,-2) to (0,2);
\draw[-,very thick] (3,-2) to (3,2);
\draw[-,very thin] (-3,2) to (0,0);
\draw[-,very thin] (0,0) to (3,-2);
\draw[-,very thin] (-3,0) to (0,-2);
\draw[-,very thin] (0,2) to (3,0);
\end{tikzpicture}
\hspace{4em}
        \scalebox{1}{\nfgame{$a_1$ $a_2$ $b_1$ $b_2$ $1$ $0$ $0$ $1$ $0$ $1$ $1$ $0$}}
		\caption{Illustration for Theorem~\ref{thm:No-NE-partySP-twoparties}. The left figure shows the voters supporting candidates of~$A$ and~$B$ in all nomination \schemes. On the right, the corresponding winners and non-winners are indicated by ``1'' and ``0'', respectively.
        }\label{fig:NonEx-twoparties}
        \end{center}
	\end{figure}
\end{proof}
    
We proceed to show the tightness of Theorem~\ref{thm:NE-existence-three-parties} in terms of the number of parties.

\begin{theorem}\label{thm:NoNE}
There exists a 1-D Euclidean election~$\E$  with four parties and maximum party size~$2$ with both single-peaked and party-aligned single-peaked preference profiles that does not admit any Nash equilibria.
\end{theorem}
\begin{proof}
	
		Let us construct an instance with four parties that is both single-peaked and party-aligned single-peaked. We do this by constructing the following 1-D Euclidean election.
        
        Let the \emph{distribution function $\phi$} of voters assign to each point~$x$ on the real line the number $\phi(x)$ of voters located at~$x$. We set the distribution function~$\phi$ of our instance such that $\phi(2)=5$, $\phi(3)=2$, $\phi(5)=6$, $\phi(8.9)=2$, $\phi(11)=7$, and $\phi(x)=0$ for all other points~$x$ on the real line.  We define the set of parties $\P=\{P_1,P_2,P_3,P_4 \}$. 
        Let $P_1=\{p_1,p'_1\}$, $P_2=\{p_2,p'_2\}$, $P_3=\{p_3\}$ and~$P_4=\{p_4\}$. The positions of the candidates are 
        as shown in Figure~\ref{fig:nonEx-line}.

        \begin{figure}
            \centering
           \scalebox{.8}{ \includegraphics[width=0.7\columnwidth]{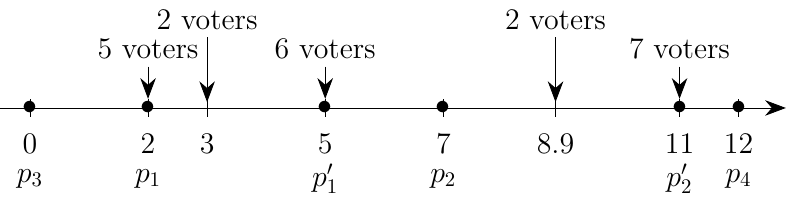}}
            \caption{Location of voters and candidates on the real line in the election constructed in Theorem~\ref{thm:NoNE}. In all figures, voters are indicated by arrows, and candidates by black circles.}
            \label{fig:nonEx-line}
        \end{figure}
        
        It is straightforward to verify that all preference profiles in this election are party-aligned single-peaked (note that voters' preferences can be formulated as strict total orders) and also single-peaked. Furthermore, we observe that parties $P_3$ and $P_4$ nominate the same candidate in each nomination \scheme, and that their candidates cannot win in any nomination \scheme. Figure~\ref{fig:W-LNoNE} shows the number of votes obtained by parties~$P_1$ and~$P_2$, as well as the winner of the resulting election, depending on the nominated candidates. 
		In each nomination \scheme, either party~$P_1$ or party~$P_2$ has a Nash deviation, so the instance admits no Nash equilibria.
		\begin{figure}[t]
           \begin{center}
			\scalebox{1}{\nfgame{$p_1$ $p'_1$ $p_2$ $p'_2$ $7$ $8$ $8$ $2$ $13$ $9$ $8$ $9$}}
			~
            \hspace{4em}
			\scalebox{1}{\nfgame{$p_1$ $p'_1$ $p_2$ $p'_2$ $0$ $1$ $1$ $0$ $1$ $0$ $0$ $1$}}
            \end{center}
			\caption{The left figure shows the numbers of voters supporting candidates of $P_1$ and $P_2$ in all possible nomination \schemes. On the right, the corresponding winners and non-winners are indicated by ``$1$'' and~``$0$'', respectively.
            }\label{fig:W-LNoNE}
		\end{figure}
%
%
\end{proof}

\section{Finding a Nash Equilibrium}
\label{sec:findNashEq}

Let us now turn our attention to the problem of finding a Nash equilibrium whenever it exists. 
As we have already shown in Section~\ref{sec:equilibrium-existence}, a Nash equilibrium may not exist even in very restricted domains involving just four parties. Formally,
we define the corresponding 
problem as follows:

\medskip
\noindent
\begin{minipage}{0.88\columnwidth}
\fbox{
    \parbox{1.08\columnwidth}{
    {\Nashexistence} 
    \begin{description}
        \item[{\bf Input:}] 
        A set~$C$ of candidates partitioned into a set~$\mathcal{P}$ of parties
         and a set~$V$ of voters with preference profile ${(\succ_v)_{v \in V}}$ over the set~$C$. 
         \item[{\bf Question:}] 
         Does $(C,\P,V,(\succ_v)_{v \in V})$ admit a Nash equilibrium?
    \end{description}
    }
    }
\end{minipage}
\medskip

It turns out that \Nashexistence\ is already $\NP$-hard for single-peaked preference profiles, as we show in Theorem~\ref{thm:find-Nash-SP-NPhard}. 
In stark contrast with this, we give a polynomial-time algorithm that not only finds a Nash equilibrium, but also decides if a given party can win in some Nash equilibrium.  We 
solve the following problem.

\medskip
\noindent
\begin{minipage}{0.88\columnwidth}
\fbox{
    \parbox{1.08\columnwidth}{
    {\findNashEq} 
    \begin{description}
        \item[{\bf Input:}] 
         A set~$C$ of candidates partitioned into a set~$\mathcal{P}$ of parties, a distinguished party~$P \in \P$, and a set~$V$ of voters with preference profile ${(\succ_v)_{v \in V}}$ over the set~$C$. 
         \item[{\bf Question:}] 
         Is there a nomination \scheme\ that is a Nash equilibrium for the instance and where $P$ is a winner in the resulting election?
    \end{description}
    }
    }
\end{minipage}
\medskip

\begin{restatable}[\linkproof{sec:proof-of-thmfindNash}]{theorem}{thmfindNashSPNPhard}    
    \label{thm:find-Nash-SP-NPhard}
    \Nashexistence\ and \findNashEq\ are $\NP$-hard for 1-D Euclidean, and hence, for single-peaked preference profiles. 
\end{restatable}

We move on to present our main result:
\begin{theorem}
    \label{thm:find-Nash-Eq-poly}
    \Nashexistence\ and \findNashEq\ are polynomial-time solvable if the preference profile is party-aligned single-peaked.
\end{theorem}

We present a polynomial-time algorithm for \findNashEq. 
Note that 
applying this algorithm for each party as the distinguished one results in a polynomial-time algorithm for \Nashexistence.

Let $(C,\P,P,V,(\prec_{v})_{v \in V})$ be our input instance for the \findNashEq\ problem.
Our algorithm uses dynamic programming and relies heavily on the structure of the preference domain. 

We start by computing a party axis~$\triangleleft_{\P}$ with respect to which the input profile $(\prec_{v})_{v \in V}$ is party-aligned single-peaked; let $P_1,\dots,P_k$ be the parties in~$\P$ ordered according to~$\triangleleft_{\P}$, and let $P=P_{\kappa}$ be the distinguished party among whose candidates we are searching for a candidate that can become a winner in some Nash equilibrium.

We first prove the following lemma.
\begin{restatable}[\linkproof{sec:proof-of-twopossvotes}]{lemma}{lemtwopossiblevotes}
\label{lem:two-possible-votes}
Suppose that the profile $(\succ_v)_{v \in V}$ is party-aligned single-peaked.
Then for each voter $v \in V$, there exist at most two parties for which $v$ might vote for in a reduced election.
Moreover these two parties must be adjacent along the party axis, and can be found in $O(|C|)$ time.
\end{restatable}
\paragraph{Partitioning the Voter Set.}
We now introduce the following sets of voters.
First, for each $i \in [k]$, let $V^{P_i}$ denote the set of those voters in~$V$ who always vote for the nominee of party~$P_i$ irrespective of the nomination \scheme. Note that $v \in V^{P_i}$ if and only if the top $|P_i|$ candidates in $v$'s preference list are exactly the candidates belonging to~$P_i$.

Second, for each $i \in [k-1]$, let $V^{(P_i,P_{i+1})}$ denote the set of voters in~$V \setminus (V^{P_i} \cup V^{P_{i+1}})$ who always vote either for the nominee of~$P_i$ or for the nominee of~$P_{i+1}$. 
By Lemma~\ref{lem:two-possible-votes}, the sets $V^{P_i}$ for $i \in [k]$ together with the sets $V^{(P_i,P_{i+1})}$ for $i \in [k-1]$ form a partitioning of the voter set~$V$, and  this partitioning can be computed in $O(|V| \cdot |C|)$ time.

Third, let us define the following sets of voters, crucial for our dynamic programming approach:
\allowdisplaybreaks
\begin{align}
\label{equ:def-Vleqi}
V_{\leq i}&=\left(\bigcup_{j=1}^i V^{P_j}  \right) \cup \left( \bigcup_{j=1}^{i-1} V^{(P_j,P_{j+1})}\right), \quad \text{ and} \\
V_{\geq i}&=\left(\bigcup_{j=i}^k V^{P_j}  \right) \cup \left( \bigcup_{j=i}^{k-1} V^{(P_j,P_{j+1})}\right).
\label{equ:def-Vgeqi}
\end{align}
Note that $V_{\leq 1} \subseteq V_{\leq 2} \subseteq \dots \subseteq V_{\leq k}=V$, and observe also that voters in~$V_{\leq i}$ may only vote for nominees in~$P_1 \cup \dots \cup P_i$ in any reduced election.
Similarly, $V=V_{\geq 1} \supseteq V_{\geq 2} \supseteq \dots \supseteq V_{\geq k}$, and voters in~$V_{\geq i}$ may only vote for nominees in~$P_i \cup \dots \cup P_k$ in any reduced election.


Assuming some hypothetical nomination \scheme\ that is a Nash equilibrium and in which $P_\kappa$ is a winner, we next guess the score~$s^\star$ obtained by the nominee of~$P_\kappa$ in such a Nash equilibrium; 
note that $s^\star \leq |V|$, so there are at most~$|V|$ possible values $s^\star$ can take. Henceforth, we treat $s^\star$ as fixed.

\paragraph{Feasible Partial Nomination Schemes.} 
For some ${i \in [k]}$ at most $\kappa$, a \emph{left-feasible partial nomination \scheme\ for~$i$} is a tuple $(c_1,\dots,c_i)$ where $c_j \in P_j$ for each $j \in [i]$, and

\smallskip
\begin{minipage}{0.96\columnwidth}
\begin{itemize}
    \item[($\alpha$)] for each $j \in [i-1]$, $c_j$ obtains at most~$s^\star$ votes in every nomination \scheme\ $(c_1,\dots,c_i,c'_{i+1},\dots,c'_k)$,\footnotemark and
    \item[($\beta$)]  for each $j \in [i-1]$,  
    if $c_j$ obtains less than~$s^\star$ votes in some nomination \scheme\ $(c_1,\dots,c_i,c'_{i+1},\dots,c'_k)$, then no candidate $c'_j \in P_j \setminus \{c_j\}$ obtains at least~$s^\star$ votes in $(c_1,\dots,c_{j-1},c'_j,c_{j+1}, \dots,c_i,c'_{i+1},\dots,c'_k)$.
\end{itemize}
\end{minipage}
\smallskip

\footnotetext{Note that the nominees $c'_{i+1},\dots,c'_k$ are irrelevant here as they do not affect how voters in~$V_{\leq i}$ vote.}

Similarly, for each index~$i \in [k]$ at least~$\kappa$, a \emph{right-feasible partial nomination \scheme\ for~$i$} is a tuple $(c_1,\dots,c_i)$ where $c_j \in P_j$ for each $j \in [i]$, and moreover, %

\smallskip
\begin{minipage}{0.95\columnwidth}
\begin{itemize}
    \item[($\alpha'$)] for each $j \in [k]\setminus [i]$, $c_j$ obtains at most~$s^\star$ votes in every nomination \scheme\ $(c'_1,\dots,c'_{i-1},c_i,\dots,c_k)$, and
    \item[($\beta'$)]  for each $j \in [k]\setminus [i]$,  
    if $c_j$ obtains less than~$s^\star$ votes in some nomination \scheme\ $(c'_1,\dots,c'_{i-1},c_i,\dots,c_k)$, then no candidate $c'_j \in P_j \setminus \{c_j\}$ obtains at least~$s^\star$ votes in $(c'_1,\dots,c'_{i-1},c_i,\dots,c_{j-1},c'_j,c_{j+1}, \dots,c_k)$.
\end{itemize}
\end{minipage}
\smallskip

\paragraph{Viable Scores and Score Pairs.}
We are now ready to define the central notions necessary for our algorithm. 
For some $i \in [k]$ and candidate~$c_i \in P_i$, an integer $s_i$ is a \emph{left-} (or \emph{right-}) \emph{viable score for~$c_i$} if there is a left- (or right-) feasible partial nomination \scheme\ where $c_i$ obtains exactly~$s_i$ votes from~$V_{\leq i}$ (from $V_{\geq i}$, respectively). 

Next, we define the concept of viability for \emph{pairs of integers}.
Given an index~$i \in \{2,\dots,k\}$, a pair $(s_{i-1},s_i)$ of integers is \emph{left-viable for a pair $(c_{i-1},c_{i}) \in P_{i-1} \times P_{i}$ of candidates} if there is a left-feasible partial nomination \scheme\ where $c_j$ obtains exactly~$s_j$ votes from~$V_{\leq i}$ for both $j \in \{i-1,i\}$. 
Notice that, consequently, some score~$s_i$ is left-viable for candidate~$c_i \in P_i$ 
if and only if there exists some $s_{i-1} \in \{0,1,\dots,s^\star\}$ and candidate~$c_{i-1} \in P_{i-1}$ such that$(s_{i-1},s_i)$ is a left-viable pair for~$(c_{i-1},c_i)$. 
The symmetric concept of \emph{right-viability} is defined analogously. 

Lemma~\ref{lem:viable-score-vs-equilibria} shows the importance of these notions.

\begin{restatable}[\linkproof{sec:proof-of-lemviablescores}]{lemma}{lemviablescorevsequilibria}
    \label{lem:viable-score-vs-equilibria}
    Given a candidate~$c_\kappa \in P_\kappa$, the following statements are equivalent: 

\smallskip
\begin{minipage}{0.96\columnwidth}
    \begin{itemize}
        \item[(a)] there exists a nomination \scheme\ 
        that is a Nash equilibrium and in which $c_\kappa$ is a winner with score~$s^\star$;
        \item[(b)] there exist non-negative  integers $s_l$ and~$s_r$ satisfying $s_l+s_r-|V^{P_\kappa}|=s^\star$ such that 
        $s_l$ is a left-viable score for~$c_\kappa$
        and $s_r$ is a right-viable score for~$c_\kappa$.
    \end{itemize}
\end{minipage}    
\end{restatable}

\paragraph{Computing Viable Score Pairs.} 
Let us now provide a dynamic programming method for computing all left-viable score pairs for each candidate pair $(c_i,c_{i+1}) \in P_i \times P_{i+1}$ increasingly for each $i \in [\kappa-1]$; computing right-viable scores can be done in a symmetric manner.
From now on, we may write viability instead of left-viability.

Notice that for a given candidate pair~$(c_1,c_2) \in P_1 \times P_2$, the only viable pair can be $(s_1,s_2)$ where
\begin{align}
\notag
s_1 &=f(c_1|c_2):=|V^{P_1}|+|\{v:v \in V^{(P_1,P_2)},c_1 \succ_v c_2\}|, \\
\label{eq:def-scorepair-init}
s_2 &= |V^{P_2}|+|\{v:v \in V^{(P_1,P_2)},c_2 \succ_v c_1\}|.
\end{align}
However, such a pair~$(s_1,s_2)$ is not necessarily viable for~$(c_1,c_2)$, because it might not lead to a Nash equilibrium. 

\begin{restatable}[\linkproof{sec:proof-of-lemviablescorepair}]{lemma}{lemNEviablescorepairi}
    \label{lem:NE-viable-scorepair-i=1}
    The pair $(s_1,s_2)$ defined as in (\ref{eq:def-scorepair-init}) is viable for a candidate pair $(c_1,c_2) \in P_1 \times P_2$ if and only if \\
    ($\dagger$) $s_1=s^\star$, or \\
    ($\ddagger$) $s_1<s^\star$ and there is no $c'_1 \in P_1$ for which $f(c'_1|c_2)\geq s^\star$.
\end{restatable}

Now, we take the general case where $2<i \leq k$ and aim to compute the viable score pairs for the  pairs in $P_{i-1} \times P_i$.
\begin{restatable}[\linkproof{sec:proof-of-lemviablescorepairrec}]{lemma}{lemNEviablescorepairrecursion}
    \label{lem:NE-viable-scorepair-recursion}
    A pair $(s_{i-1},s_i)$ of non-negative integers with $s_{i-1}\leq s^\star$ is viable for $(c_{i-1},c_i) \in P_{i-1} \times P_i$ if and only if 
    there exist integers $s'_{i-1},s_{i-2} \in \{0,1,\dots,s^\star\}$ and candidate $c_{i-2} \in P_{i-2}$ such that 

    \smallskip
    \begin{minipage}{0.96\columnwidth}
    \begin{itemize}
            \item[(i)] $(s_{i-2},s'_{i-1})$ is a viable pair for~$(c_{i-2},c_{i-1})$,
            \item[(ii)] $s_{i-1}=s'_{i-1}+|\{v:v \in V^{(P_{i-1},P_i)}, c_{i-1} \succ_v c_i\}|$,
            \item [(iii)] $s_i=|\{v:v \in V^{(P_{i-1},P_i)}, c_i \succ_v c_{i-1}\}|+|V^{P_i}|$, and 
            \item[(iv)] either ($\dagger$) $s_{i-1}=s^\star$, 
            or ($\ddagger$) there is no $c'_{i-1} \in P_{i-1}$ for which  $f(c'_{i-1}|c_i,c_{i-2})\geq s^\star$ where 
    \end{itemize}
    \end{minipage}
    \begin{align*} 
    f(&c'_{i-1} |c_i,c_{i-2}) = |\{v:v \in V^{(P_{i-1},P_i)}, c'_{i-1}\succ_v c_i\}|\\
    & + |V^{P_{i-1}}| +
    \{v:v \in V^{(P_{i-2},P_{i-1})}, c'_{i-1}\succ_v c_{i-2}\}|     .
    \end{align*}    
\end{restatable}

Using dynamic programming based on Lemmas~\ref{lem:NE-viable-scorepair-i=1} and~\ref{lem:NE-viable-scorepair-recursion}, we can compute all left-viable score pairs for each candidate pair $(c_{i-1},c_i) \in P_{i-1} \times P_i$ increasingly for each $i=2,\dots,\kappa$; note that it suffices to compute score pairs~$(s_{i-1},s_i)$ where $s_{i-1}$ and 
$s_i$ are between~0 and~$s^\star$. 
Hence, we need to compute a binary variable for each  $s_{i-1},s_i \in \{0,\dots,s^\star\}$ and each $c_{i-1} \in P_{i-1}$ and $c_i \in P_i$, for $i=2,\dots,\kappa$, describing whether $(s_{i-1},s_i)$ is viable for~$(c_{i-1},c_i)$ or not; we can determine these values for $i=2$ using Lemma~\ref{lem:NE-viable-scorepair-i=1}, and we can compute them for increasing values of~$i$ using Lemma~\ref{lem:NE-viable-scorepair-recursion}. 

This allows us to compute all left--viable scores for each candidate in~$P_\kappa$. Computing right-viable scores can be done in a symmetric manner. Finally, we Lemma~\ref{lem:viable-score-vs-equilibria} to either find a candidate~$c_\kappa \in P_\kappa$ that wins in some Nash equilibrium with score~$s^\star$ or conclude correctly that no such candidate exists. 


\paragraph{Running Time.} 
Computing the party axis can be done in polynomial time, according to Theorem~\ref{thm:recognition-poly}.
Guessing the winning score~$s^\star$ of the distinguished party~$P_\kappa$ adds a factor of~$|V|$ to the running time necessary for computing the viable scores.
Checking the conditions of Lemmas~\ref{lem:NE-viable-scorepair-i=1} and~\ref{lem:NE-viable-scorepair-recursion} can clearly be done in polynomial; hence, the running time of our algorithm is polynomial in the input size.
This finishes the proof of Theorem~\ref{thm:find-Nash-Eq-poly}.

\section{Finding a Possible or Necessary President}
\label{sec:pospres}

Next, we show that the following problems are polynomial-time solvable 
in the domain of party-aligned single-peaked profiles, by an adaptation of our algorithm for Theorem~\ref{thm:find-Nash-Eq-poly}:

\medskip
\noindent
\begin{minipage}{0.88\columnwidth}
\fbox{
        \parbox{1.08\columnwidth}{{\textsc{Possible} (or \textsc{Necessary) President}}
        \begin{description}
            \item[{\bf Input:}] 
            A set~$C$ of candidates partitioned into a set~$\mathcal{P}$ of parties, a distinguished party~$P \in \P$, and a set~$V$ of voters with preference profile ${(\succ_v)_{v \in V}}$ over the set~$C$.
            \item[{\bf Question:}] 
            Is $P$ a winner in the election resulting from \emph{some} (or \emph{all}, respectively) nomination \scheme(s)?
        \end{description}
    }
    }
\end{minipage}
\medskip

Faliszewski et al.~(\citeyear{faliszewski2016}) proved that \pospres\ is $\NP$-complete for single-peaked preferences, while \necpres\ is $\coNP$-complete in general but becomes polynomial-time solvable for single-peaked preferences. 
They also showed that \pospres\ is polynomial-time solvable if the preferences are single-peaked with respect to an axis along which the candidates of each party are ordered consecutively, that is, if preferences are single-peaked and also party-aligned single-peaked.

The following theorem extends the tractability results of Faliszewski et al.~(\citeyear{faliszewski2016}) to the domain of party-aligned single-peaked preference profiles.

\begin{restatable}[\linkproof{sec:app-pospres}]{theorem}{thmPossiblePrespoly}
    \label{thm:Possible-Pres-poly}
    The \textsc{Possible} and \textsc{Necessary President} problems are polynomial-time solvable if the preference profile is party-aligned single-peaked.
\end{restatable}

\begin{proofsketch}
    The algorithm for \findNashEq\ presented in Section~\ref{sec:findNashEq} can be adapted in a straightforward way to solve the \pospres\ problem; in fact, this is a simplification, since we only need to ensure that the distinguished party is a winner in some nomination scheme~$\vec{c}$, but we do not require $\vec{c}$ to be a Nash equilibrium. 
    
    To solve the \necpres\ problem, we apply a slightly modified variant of our algorithm for \pospres\ that not only ensures that the distinguished party~$P^+$ is a winner in some nomination scheme, but also guarantees that some other party~$P^-$ is \emph{not} a winner; running this algorithm for each party~$P^+$ other than~$P^-$, we can decide if $P^-$ wins in all reduced elections.
\end{proofsketch}

\section{Discussion}

\begin{table}
	\begin{center}
		\renewcommand{\arraystretch}{1.5}
		\renewcommand{\tabcolsep}{5pt}

	\scalebox{.9}{	
    \begin{tabular}{l c c c}
			& {\bf\textsc{Equi Pres}} & {\bf \textsc{Pos Pres}} & {\bf\textsc{Nec Pres}}  \\ \toprule
			{\bf SP \, \& \, PASP} & \textsc{P} {\footnotesize[\Shcref{thm:find-Nash-Eq-poly}]} 
            & \textsc{P} {\footnotesize[\Shcref{thm:Possible-Pres-poly}]} &  
            \textsc{P}             {\footnotesize[F'16]} 
            \\
			{\bf PASP } &  \textsc{P} {\footnotesize[\Shcref{thm:find-Nash-Eq-poly}]} 
            & 
            \textsc{P} {\footnotesize[\Shcref{thm:Possible-Pres-poly}]}
            & 
            \textsc{P} {\footnotesize[\Shcref{thm:Possible-Pres-poly}]} 
            \\
            {\bf SP} 
            &  \textsc{NP}-c {\footnotesize[\Shcref{thm:find-Nash-SP-NPhard}]} 
            &  \textsc{NP}-c {\footnotesize[F'16]}
            & \textsc{P} {\footnotesize[F'16]} \\
			\bottomrule
		\end{tabular}
        }
			\end{center}
		
		\caption{Summary of our algorithmic results. 
        SP (and PASP) stands for single-peaked (party-aligned single-peaked) preferences. [F'16] marks results by~\citet{faliszewski2016}.}
		\label{tab:overview}
	\end{table}
We have explored computational questions about strategic candidate nomination by parties in a model where voters' preferences are party-aligned single-peaked, with a focus on the concept of Nash equilibria; see Table~\ref{tab:overview} for an overview of our algorithmic results. Studying different restrictions on the preference domain or voting rules other than Plurality are a natural direction for further research in the same model.
More generally, extensions of our model that deal with possible interactions between the parties also seem interesting. 
For example, one could consider what a group of parties can achieve if they were to coordinate their nominations, with a coalition $\mathcal{C} \subseteq \mathcal{P}$ being \emph{powerful} if there exists a tuple $T_\mathcal{C} = (c_i)_{P_i \in \mathcal{C}}$  of candidates such that for every nomination \scheme\ extending $T_\mathcal{C}$, the Plurality winner is in $T_\mathcal{C}$. 
This definition generalizes the notion of a \emph{necessary president} (i.e., a powerful singleton coalition), and we would like to know under what conditions we can determine it. Similar generalisations can be attempted for possible president, as well as a possible or necessary president in an equilibrium.

Understanding strategic nominations in district-based elections is another interesting direction, where the preference structure could improve equilibrium computation results. 
Finally, extensive interaction is worth exploring, where parties hold primaries at different times, with more nuanced (e.g., subgame-perfect) equilibria.

\section{Acknowledgments}
This work is supported by the European Union (ERC, PRAGMA, 101002854). Views and opinions expressed are however those of the author only and do not necessarily reflect those of the European Union or the European Research Council. Neither the European Union nor the granting authority can be
held responsible for them.  Grzegorz Lisowski acknowledges support by the European Union under the Horizon Europe project 
{Perycles} (Participatory Democracy that Scales).
Ildik\'o Schlotter was supported by the Hungarian Academy of Sciences under its Momentum Programme (LP2021\nobreakdash-2) and its J{\'a}nos Bolyai Research Scholarship.

\raisebox{-0.2\height}{\includegraphics[width=0.4\columnwidth]{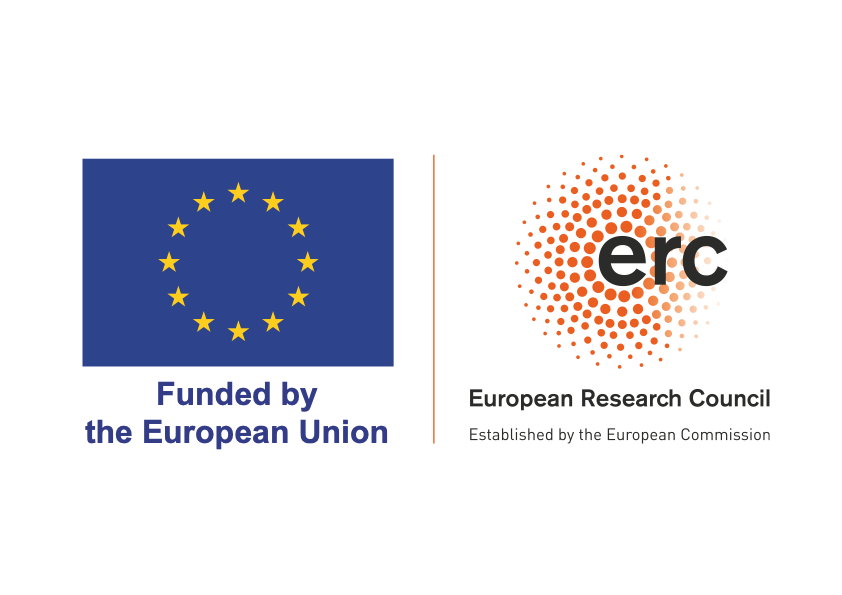}}
\hspace{4em}
\raisebox{0.4\height}{\includegraphics[width=0.5\columnwidth]{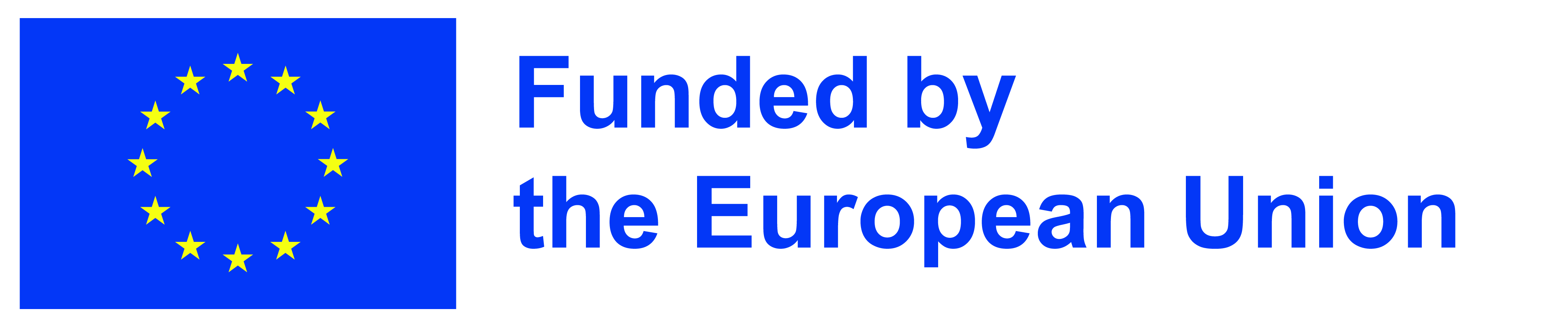}}

\bibliographystyle{ACM-Reference-Format} 
\bibliography{file}

\newpage

\ifshort

\else
\clearpage
\begin{appendices}
\section{Missing proofs from Section~\ref{sec:recognition}}

\subsection{The proof of Lemma~\ref{lem:vote_party_single_peaked}}
\label{sec:proof-of-lemvotePSP}

\lemvotepartysinglepeaked*
\begin{proof}
First, we prove that if the vote $\succ_v$ is party-aligned single-peaked according to $\triangleleft_{\mathcal{P}}$, then conditions (a)--(c) are met. 

To see that~(a) holds, assume $P_j \notin \{P_1,P_m\}$. Then for every perceived axis~$\triangleleft_v$ compatible with~$\triangleleft_\P$ we have $h_{j-1} \triangleleft_v l_j \triangleleft_v h_{j+1}$.
Hence, if $h_{j-1} \succ_v l_j$ (which is equivalent with $l_j \not \succ_v h_{j-1}$), then the single-peakedness of $\succ_v$ with axis~$\triangleleft_v$ implies that $l_j \succ_v h_{j+1}$, proving~(a).

To see condition~(b), consider two parties $P_i$ and~$P_{i+1}$ for $j < i$. Then for every perceived axis~$\triangleleft_v$ compatible with~$\triangleleft_\P$ we have $h_j \triangleleft_v l_i \triangleleft_v h_{i+1}$ for the top candidate~$h_j$. Thus, due to the single-peakedness of~$\succ_v$ with axis $\triangleleft_v$, we know that $h_j \succ_v l_i$ implies $l_i \succ_v h_{i+1}$. 
%
The argument for condition~(c) is 
symmetric to that for condition~(b).

We move on to prove that if conditions (a)--(c) are satisfied, then the vote~$\succ_v$ is party-aligned single-peaked under~$\triangleleft_{\mathcal{P}}$. We show how to construct a total ordering $\triangleleft_v$ of the set~$C$ of candidates that is compatible with the party ordering $\triangleleft_{\mathcal{P}}$ and for which vote $\succ_v$ is single-peaked under $\triangleleft_v$.


If $P_j$ is the first party in $\triangleleft_{\mathcal{P}}$, i.e., $j=1$, we identify the highest-ranked candidate of party $P_2$, that is, $h_2$. We place every candidate from party $P_j = P_1$ ranked by~$v$ lower than~$h_2$ on the left of the top candidate~$h_1$ in decreasing order of their rank according to~$\succ_v$ (i.e., setting $c' \triangleleft_v c$ whenever $c \succ_v c'$ for any two such candidates~$c$ and~$c'$). Next, we place the top candidate~$h_1$, and then the remaining candidates from~$P_1$ (those preferred to~$h_2$) in increasing order of their rank according to $\succ_v$. 
Condition~(b) guarantees that the remaining candidates from parties $P_2, P_3, \ldots, P_m$ can be placed in $\triangleleft_v$ in increasing order of their rank according to~$\succ_v$ so that $\triangleleft_v$ becomes compatible with~$\triangleleft_{\mathcal{P}}$. The case when $P_j=P_m$ is symmetric.

Suppose now that $P_j$ is neither the first nor the last party according to $\triangleleft_{\mathcal{P}}$. Then condition~(a) guarantees that at least one of the highest-ranked candidates of parties $P_{j-1}, P_{j+1}$, that is, one of $h_{j-1}$ and $h_{j+1}$, is ranked lower than $l_j$ by~$v$. Breaking symmetry, we may assume that $l_j \succ_{v} h_{j-1}$. Similarly to the previous case, we place every candidate from party $P_j$ ranked lower than $h_{j+1}$ on the left of the top candidate~$h_j$ in decreasing order of their rank according to $\succ_v$. Next, we place the top candidate~$h_j$, and then the remaining candidates from~$P_j$ in increasing order of their rank according to~$\succ_v$. Conditions~(b) and~(c) guarantee that the candidates from the remaining parties can be placed in the total order~$\triangleleft_v$ so that it is compatible with~$\triangleleft_\P$.
\end{proof}

\subsection{The proof of Lemma~\ref{lem:bottom_parties}}
\label{sec:proof-oflembottomparties}
\lembottomparties*
\begin{proof}
If party $P_i$ is neither in the leftmost nor in the rightmost position in $\triangleleft_{\mathcal{P}}$, then for the perceived axis~$\triangleleft_v$ of voter~$v$,  which must be compatible with~$\triangleleft_{\mathcal{P}}$, there exist candidates~$a$ and~$b$ such that $a \triangleleft_v c \triangleleft_v b$. However, then $a \succ_{v} c$ and $b \succ_{v} c$ contradict the single-peakedness of~$\succ_v$  with the perceived axis~$\triangleleft_v$.
\end{proof}

\subsection{The proof of Lemma~\ref{lem:party_fixed_placement}}
\label{sec:proof-of-partyfixedplacement}

\lempartyfixedplacement*
\begin{proof}
Since the set of all candidates from parties in $R \cup L$ is not the suffix of the vote $\succ_v$, it holds, by definition, that $b \succ_v a$. Moreover, for every party in $P_i \in \mathcal{P} \setminus (L \cup R)$ different from~$P_a$, it holds that $v$ prefers every candidate of~$P_i$ to~$a$. 

Assume that $P_a$ is not placed in~$\triangleleft_\P$ as claimed. 
This means that there exists a party $P_i \in \mathcal{P} \setminus (L \cup R \cup \{P_a\})$  such that $P_a$ is placed between $P_i$ and $P_b$ according to $\triangleleft_{\mathcal{P}}$. 
However, both $P_b$ and $P_i$ have at least one candidate that is ranked strictly higher than $a \in P_a$, which contradicts our assumption that $\succ_v$ is party-aligned single-peaked with axis~$\triangleleft_\P$.
\end{proof}

\subsection{The proof of Lemma~\ref{lem:reducing_problem}}
\label{sec:proof-of-reducingproblem}

\lemreducingproblem*
\begin{proof}
    Consider a voter~$v \in V$, and let $\triangleleft^Q_v$ be a perceived axis over~$Q$  compatible with $\triangleleft_Q$ such that the restriction of~$\succ_v$ to~$Q$ is single-peaked w.r.t.~$\triangleleft^Q_v$. 
    Let also $\triangleleft^C_v$ be a perceived axis compatible with $\triangleleft_\P$ such that $\succ_v$ is single-peaked w.r.t.~$\triangleleft^C_v$. Define now a perceived axis $\triangleleft_v$ by combining these two such that 
    \begin{itemize}
        \item for each $c,c' \in C \setminus Q$ we set $c \triangleleft_v c'$ exactly if $c \triangleleft^C_v c'$;
        \item for each $c,c' \in Q$ we set $c \triangleleft_v c'$ exactly if $c \triangleleft^Q_v c'$;
        \item for each $c \in C \setminus Q$ and $c' \in Q$, we set $c \triangleleft_v c'$ exactly if $c \in P_1 \cup \dots P_i$.
    \end{itemize}
    It is straightforward to check that $\triangleleft_v$ is compatible with the ordering $P_1 , \dots, P_i , Q_1, \dots,  Q_k , P_j, \dots, P_m$. 
    Hence, it remains to show that $\succ_v$ is single-peaked w.r.t.~$\triangleleft_v$.
    Consider candidates $a,b,$ and~$c$ such that $a \triangleleft_v b \triangleleft_v c$ and $a \succ_v b$; we aim to show that $b \succ_v c$. 
    
    Now, if $\{a,b,c\}$ is a subset of $C \setminus Q$ (or of~$Q$), then $b \succ_v c$ follows because the vote  $\succ_v$ (or its restriction to~$Q$) is single-peaked w.r.t.\ the perceived axis~$\triangleleft^C_v$ (or $\triangleleft^Q_v$, respectively). Hence, we may assume that $|Q \cap \{a,b,c\}| \in [2]$. 
    
     First, if $b \in Q$ and $c \in C \setminus Q$, then $b \succ_v c$ holds, because $C \setminus Q$ forms a suffix of~$\succ_v$.
     Second, if $\{a,b \} \subseteq C \setminus Q$ or $\{b,c\} \subseteq C \setminus Q$, then our definition of~$\triangleleft_v$ yields $a \triangleleft_v^C b \triangleleft_v^C c$, so the single-peakedness of~$\succ_v$ w.r.t.~$\triangleleft_v^C$ means that $a \succ_v b$ indeed implies $b \succ_v c$.
     The only remaining possibility is that $a \in C \setminus Q$ and $b \in Q$. However, in this case $a \succ_v b$ cannot happen, because $C \setminus Q$ forms a suffix of~$\succ_v$. Hence, in all possible cases, we obtain $b \succ_v c$, as required.
\end{proof}

\section{Missing proofs from Section~\ref{sec:findNashEq}}

\subsection{The proof of Theorem~\ref{thm:find-Nash-SP-NPhard}}
\label{sec:proof-of-thmfindNash}

\thmfindNashSPNPhard*
\begin{proof}
    We reduce from the $\NP$-hard \textsc{Exact 3-Cover} problem~\cite{gar-joh:b:int} whose input is a set $U=\{u_1,\dots,u_{3n}\}$ and a family~$\S=\{S_1,\dots,S_m\}$ of {size-3} subsets of~$U$, and the task is decide if there are $n$ sets in~$\S$ whose union is~$U$.

    Let us create an instance 
    $I$ of \Nashexistence\ as follows.
    We create $m$ \emph{triple segments} and $m-n$ \emph{auxiliary segments} on the real line. Each triple segment corresponds to a set~$S_j \in \S$ and has length $L=4(m-n)+7$, so that the segment corresponding to~$S_j$ is located at the interval $[(j-1)L,jL]$. 
    Each of the $m-n$ auxiliary segments has length~$14$, with the $i^\textrm{th}$  located at $[K+14(i-1),K+14i]$ where $K=Lm+1$ is the start of the first auxiliary segment. There will be two voters located in each triple segment, and four voters in each auxiliary segment.

    \begin{figure}
        \centering
\scalebox{.8}{        \includegraphics[width=   0.7\columnwidth]{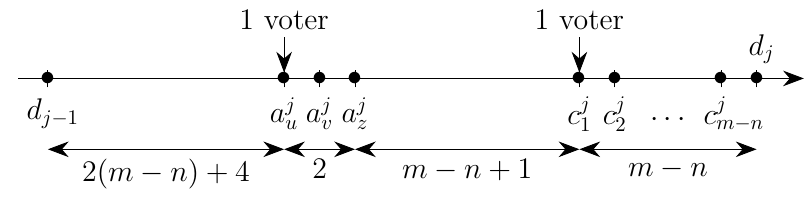}}
        \caption{The location of voters and candidates within the triple segment corresponding to~$S_j=\{a_u,a_v,a_z\} \in \S$.}
        \label{fig:triple_seg}
    \end{figure}

    \begin{figure}
        \centering
      \scalebox{.8}{ \includegraphics[width=0.7\columnwidth]{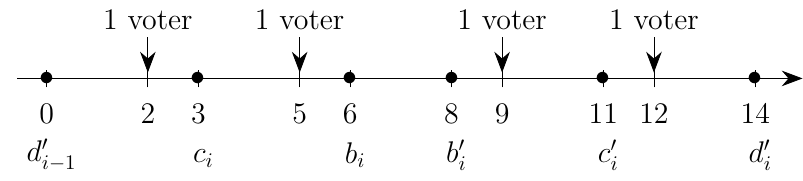}}
        \caption{The location of voters and candidates within the $i^\textrm{th}$ auxiliary segment for some $i \in [m-n]$.
        The numbers indicating locations are relative to the start of the segment.}
        \label{fig:aux_seg}
    \end{figure}
    
    We create a party $A_u=\{a_u^j:u \in S_j \in \S\}$ for each element ${u \in U}$, and 
    two parties $B_i=\{b_i,b'_i\}$ and $C_i=\{c_i,c'_i,c_i^1,\dots,c_i^m\}$  for each index $i \in [m-n]$. We also create dummy candidates $d_1,\dots,d_m$  
    and $d'_0,\dots,d'_{m-n}$, each of them contained in a singleton party, and located at the beginning or the end of some segment. 
    The locations of all candidates and all voters along the real line are given in Figures~\ref{fig:triple_seg} and~\ref{fig:aux_seg}. It is straightforward to check that voters' preferences are strict linear orders over the candidate set, so the constructed instance is within the single-peaked preference domain.

    \smallskip
    We claim that $(U,\S)$ is a yes-instance of \textsc{Exact 3-Cover} if and only if the constructed instance~$I$ admits a Nash equilibrium.
    \smallskip

    Notice that any candidate~$a_u^j$ can only obtain the vote of at most one voter, namely the ``left'' voter within the $j$-th triple segment, because the ``right'' voter in that segment prefers~$d_j$ to~$a_u^j$. This will imply that no party~$A_u$, $u \in U$, can win in any reduced election. 

    Assume that $\S' \subseteq \S$ contains exactly  $n$ sets whose union is~$U$. Let us fix a bijection~$\varphi:[m-n] \to \{j:S_j \in \S \setminus \S'\}$. Consider the nomination \scheme\ $\vec{c}$ where
    party~$A_u$ nominates the candidate $a_u^j$ for which $u \in S_j \in \S'$,
    party~$C_i$ nominates $c_i^{\varphi(i)}$, and all other nominations are fixed arbitrarily. 
    Notice that each party~$B_i$ obtains two votes from the two middle voters in the $i^\textrm{th}$ auxiliary segment, because no candidate $c_i$ or $c'_i$ is nominated.
    Also, each party~$C_i$ obtains two votes from the two voters within the triple segment corresponding to~$S_{\varphi(i)}$, because no other party has a nominee within that segment (except for the parties of dummy candidates). 
    It is also clear that each dummy candidate obtains at most two votes.
    Thus, all parties $B_i$ and $C_i$, $i \in [m-n]$, are winners, and furthermore, no Nash deviation is possible for any party~$A_u$ as each such party can obtain at most one vote. Hence, $\vec{c}$ is a Nash equilibrium.

    Assume now that there is a nomination scheme $\vec{c}$ that is a Nash equilibrium. It is not hard to check that each nominee in a segment can obtain at most two votes. Moreover, no candidate~$c_i$ or~$c'_i$ can be present in~$\vec{c}$, because whenever $c_i$ or~$c'_i$ is nominated, there always exists a Nash deviation either for $B_i$ or for~$C_i$, as shown in Figure~\ref{fig:aux-deviations}. In fact, every party~$C_i$ has to be a winner in~$\vec{c}$, as otherwise it has a Nash deviation by nominating~$c_i$ or~$c'_i$ according to Figure~\ref{fig:aux-deviations}. Consider the nominees $c_i^j$ of parties~$C_i$ for $i \in [m-n]$, and notice that each such nominee~$c_i^j$ must be the only nominee located in the triple segment~$S_j$, as otherwise $c_i^j$ (or some other nominee $c_{i'}^j$) would not obtain the two votes necessary for becoming a winner. Hence, all $3n$ nominees of the parties~$A_u$, $u \in U$ must be located in the $n$ triple segments that do not contain any candidate nominated by some party~$C_i$. This means that the $n$ sets in~$\S$ corresponding to these segments cover~$U$, as required, showing the correctness of our reduction and proving that \Nashexistence\ is $\NP$-hard.

    \begin{figure}[t]
		\begin{center}
		\scalebox{1}{\nfgame{$c_i$ $c'_i$ $b_i$ $b'_i$ $1$ $2$ $2$ $1$ $2$ $1$ $1$ $2$}}
        \hspace{4em}
		\scalebox{1}{\nfgame{$c_i$ $c'_i$ $b_i$ $b'_i$ $0$ $1$ $1$ $0$ $1$ $0$ $0$ $1$}}        
		\caption{The left figure shows the numbers of voters voting for candidates of $C_i$ and $B_i$ in all nomination \schemes\ where $C_i$ nominates $c_i$ or $c'_i$. On the right, the corresponding winners and non-winners among these parties are indicated by ``$1$'' and~``$0$'', respectively.
        \label{fig:aux-deviations}
        }\end{center}        
	\end{figure}

    To show the $\NP$-hardness of \findNashEq, we can modify the reduction as follows: we add a single new candidate~$c^\star$ forming its own singleton party, and placed sufficiently far away from every other candidate; we further add two voters who are placed at the same place as~$c^\star$ and thus always vote for~$c^\star$. Then $c^\star$ is a winner in some Nash equilibrium for the modified instance~$I'$ if and only if there exists a Nash equilibrium for the original instance~$I$ which in turn happens if and only if $(U,\S)$ is a yes-instance of \textsc{Exact 3-Cover}. This proves the $\NP$-hardness of  \findNashEq.    
\end{proof}

\subsection{The proof of Lemma~\ref{lem:two-possible-votes}}
\label{sec:proof-of-twopossvotes}

\lemtwopossiblevotes*
\begin{proof}
    Let party~$A$ contain the top candidate for voter~$v$, and let $a$ and $a'$ be the two extremal candidates of party~$A$ according to the perceived axis~$\triangleleft_v$ for~$v$, so that each candidate of party~$A$ appears between $a$ and~$a'$ along the axis~$\triangleleft_v$. Let us further break symmetry by assuming $a \succ_v a'$ and $a' \triangleleft_v a$ (otherwise, we may reverse our axes or rename the candidates). 
    Define $b$ as the candidate directly following~$a$ along~$\triangleleft_v$ so that $a' \triangleleft_v a \triangleleft_v b$. 
    Let $B$ the party containing~$b$; then party-aligned single-peakedness yields $A \triangleleft_{\P} B$, with no party placed between~$A$ and~$B$ along~$\triangleleft_\P$.
    
    We claim that in all nomination \schemes, $v$'s most-preferred candidate belongs to party~$A$ or party~$B$. 
    Indeed: since $A$ contains $v$'s top candidate, 
    we know that for each party~$C$ with $A \triangleleft_{\P} B \triangleleft_{\P} C$, $v$ prefers every candidate of~$B$ to each candidate of~$C$; similarly,  
    for each party~$C$ with $C \triangleleft_{\P} A \triangleleft_{\P} B$, candidate~$a'$ and, hence, all candidates of~$A$, are preferred by~$v$ to each  candidate of~$C$. In either case, the nominee of~$C$ in any nomination \scheme\ is less preferred by~$v$ than the nominee of either~$A$ or of~$B$, and hence cannot obtain $v$'s vote.

    Furthermore, if there is a candidate outside~$A$ that $v$ prefers to some candidate in~$A$, then the most-preferred such candidate must be exactly~$b$; this yields an efficient way to determine~$A$ and~$B$ by simply going through the top $|A|+1$ candidates for~$v$.
\end{proof}

\subsection{The proof of Lemma~\ref{lem:viable-score-vs-equilibria}}
\label{sec:proof-of-lemviablescores}

\lemviablescorevsequilibria*
\begin{proof}
    Assume first that (a) holds. 
    Define $s_l$ and $s_r$ as the number of votes that $c_\kappa$ obtains from voters in~$V_{\leq \kappa}$ and in~$V_{\geq \kappa}$, respectively, in the nomination \scheme\ $\vec{c}=(c_1,\dots,c_k)$. Recall that $V_{\leq \kappa} \cap V_{\geq \kappa}=V^{P_\kappa}$ and thus $s^\star=s_r+s_l-|V^{P_\kappa}|$.
    It is clear that $\vec{c}$ is left-feasible: condition~($\alpha)$ holds because no nominee obtains more than~$s^\star$ votes in~$\vec{c}$, and condition~($\beta$) holds since $\vec{c}$ is a Nash equilibrium and hence admits no Nash deviation.

    Assume now that (b) holds. Let a left- and a right-feasible partial nomination \scheme\ $(c_1,\dots,c_\kappa)$ and $(c_\kappa,\dots,c_k)$, respectively. for which  $c_\kappa$ obtains $s_l$ votes from~$V_{\leq \kappa}$ in the former and $s_r$ votes from~$V_{\geq \kappa}$ votes in the latter. Then $c_\kappa$ is a winner in the nomination \scheme\ $\vec{v}=(c_1,\dots,c_\kappa,\dots,c_k)$ with score~$s^\star$ due to~($\alpha$), and 
    $\vec{v}$ is a Nash equilibrium due to~($\beta$).
\end{proof}

\subsection{The proof of Lemma~\ref{lem:NE-viable-scorepair-i=1}}
\label{sec:proof-of-lemviablescorepair}
\lemNEviablescorepairi*
\begin{proof}
    First, if ($\dagger$) or ($\ddagger$) holds, then $(c_1,c_2)$ is a partial nomination \scheme\ satisfying both ($\alpha$) and ($\beta$) in the definition of (left-)feasibility.
    Second, if neither ($\dagger$) nor ($\ddagger$) holds, then  either $s_1>s^\star$ in which case 
    $(c_1,c_2)$ violates~($\alpha$), 
    or $s_1<s^\star$ but there exists some $c'_1 \in P_1$ that can obtain at least~$s^\star$ votes in every nomination scheme containing~$c'_1$ and~$c_2$ (and thus lead to a Nash deviation) in which case $(c_1,c_2)$ violates~($\beta$).
\end{proof}

\subsection{The proof of Lemma~\ref{lem:NE-viable-scorepair-recursion}}
\label{sec:proof-of-lemviablescorepairrec}

\lemNEviablescorepairrecursion*
\begin{proof}
    Let $\vec{c}=(c_1,\dots,c_{i-2},c_{i-1},c_i)$ be a left-feasible partial nomination \scheme\  for~$i$ where $c_j$ obtains $s_j$ votes for $j \in \{i-1,i\}$ from voters in~$V_{\leq i}$, witnessing the viability of~$(s_{i-1},s_i)$ for $(c_{i-1},c_i)$. 
    Let $s_{i-2}$ and $s'_{i-1}$ be the votes that $c_{i-2}$ and~$c_{i-1}$, respectively, obtain from voters in~$V_{\leq i-1}$. It is then straightforward to verify that (i)--(iii) hold. Moreover, if $s_{i-1}<s^\star$, then (as there can be Nash deviation for~$\vec{c}$ by party~$P_{i-1}$ due to condition~($\beta)$ in the definition of left-feasibility) there is no $c'_{i-1} \in P_{i-1}$ that obtains at least $s^\star$ votes when nominated together with $c_{i-2}$ and~$c_i$---which is equivalent with the condition $f(c'_{i-1}|c_i,c_{i-2}) \geq s^\star$. This means that (iv) is satisfied as well.

    Assume now that integers~$s_{i-2},s'_{i-1} \in \{0,1,\dots,s^\star\}$ and candidate~$c_{i-2} \in P_{i-2}$ satisfy conditions (i)--(iv), and let $\vec{c}_-=(c_1,\dots,c_{i-2},c_{i-1})$ be a left-feasible partial nomination \scheme\ for~$i-1$ witnessing~(i). Define the partial nomination \scheme\
    $\vec{c}=(c_1,\dots,c_{i-2},c_{i-1},c_i)$. Then $\vec{c}$ satisfies ($\alpha$), due to the left-feasibility of~$\vec{c}_-$ and because $c_{i-1}$ obtains $s_{i-1}\leq s^\star$ votes in~$\vec{c}$ from voters in $V_{\leq i}$ due to~(ii). 
    Additionally, $\vec{c}$ also satisfies ($\beta$) due to~(iv) and the left-feasibility of~$\vec{c}_-$. 
    Thus, $\vec{c}$ is left-feasible for~$i$, and since  
    $c_i$ obtains $s_i$ votes from voters in~$V_{\leq i}$ due to~(iii), $\vec{c}$ witnesses the viability of~$(s_{i-1},s_i)$ for~$(c_{i-1},c_i)$. 
\end{proof}

\section{The missing proof of Theorem~\ref{thm:Possible-Pres-poly}}
\label{sec:app-pospres}
This section is dedicated to prove the following result.

\thmPossiblePrespoly*

\subsection{Solving the \pospres\ problem}
We first deal with the \pospres\ problem.
Let $(C,\P,P,V,(\prec_{v})_{v \in V})$ be our input instance for the \pospres\ problem.
Our algorithm is essentially a simplification of the algorithm described in Section~\ref{sec:findNashEq} for the \findNashEq\ problem: we simply need to remove the constraint that ensures that $P$ is a winner in a Nash equilibrium. For completeness, we give a detailed proof here.

\smallskip

Again, we start by computing a party axis~$\triangleleft_\P$, and let $P_1,\dots,P_k$ be the parties in~$\P$ ordered according to~$\triangleleft_\P$, with $P_\kappa=P$ being the distinguished party. 
We will re-use Lemma~\ref{lem:two-possible-votes} and the consequent partitioning of the voters introduced in Section~\ref{sec:findNashEq}; in particular, we will rely on the notation for voter sets~$V_{\leq i}$ and~$V_{\geq i}$, as defined in~(\ref{equ:def-Vleqi})--(\ref{equ:def-Vgeqi}).

\smallskip
After computing the party axis, our algorithm next guesses the 
score~$s^\star$ obtained by the nominee of the distinguished party~$P_\kappa$ in some (hypothetical) reduced election where $P_\kappa$ is a winner; we fix~$s^\star$ for the remainder.

\paragraph{PP-feasible partial nomination \schemes.}
We now adapt the feasibility concepts introduced in Section~\ref{sec:findNashEq} to fit our current purposes.
Given some index~$i \in [k]$ at most~$\kappa$, a 
\emph{left-PP-feasible partial nomination \scheme\ for~$i$} is an $i$-tuple $(c_1,\dots,c_i)$ of candidates such that 
$c_j \in P_j$ for each $j \in [i]$, and moreover, for each $j \in [i-1]$ party~$P_j$ obtains at most~$s^\star$ votes 
in every nomination \scheme\ $(c_1, \dots, c_i, c'_{i+1},\dots,c'_k)$ 
from voters in~$V_{\leq i}$. 
The analogous notion of \emph{right-PP-feasible partial nomination \scheme} is defined symmetrically.

\paragraph{PP-viable scores.}
We next adapt our notion of viable scores from Section~\ref{sec:findNashEq} in a straightforward manner.
For some index~$i \in [k]$ and candidate~$c_i \in P_i$, we let $s_i$ be a \emph{left- (\textup{or} right-) PP-viable score for~$c_i$} if there is a left- (or right-) PP-feasible partial nomination \scheme\ where $c_i$ obtains exactly~$s_i$ votes from~$V_{\leq i}$ (from $V_{\geq i}$, respectively). 

The next lemma is a direct analog of Lemma~\ref{lem:viable-score-vs-equilibria}, formulating the solvability of our instance for \pospres\ based on the concepts of left- and right-PP-viability.
 \begin{lemma}
    \label{lem:viable-score-vs-possible-winner}
    Given a candidate~$c_\kappa \in P_\kappa$, the following statements are equivalent: 

\smallskip
\begin{minipage}{0.96\columnwidth}
    \begin{itemize}
        \item[(a)] there exists a nomination \scheme\ $(c_1,\dots,c_\kappa,\dots,c_k)$ in which $c_\kappa$ is a winner with score~$s^\star$;
        \item[(b)] there exist non-negative  integers $s_l$ and~$s_r$ satisfying $s_l+s_r-|V^{P_\kappa}|=s^\star$ such that 
        $s_l$ is a left-PP-viable score for~$c_\kappa$
        and $s_r$ is a right-PP-viable score for~$c_\kappa$.
    \end{itemize}
    \end{minipage}
\end{lemma}

\begin{proof}
     Assume first that (a) holds, and let $\E$ be the resulting election. Note that $V_{\leq i} \cap V_{\geq i}=V^{P_i}$. 
     Let $s_l$ and $s_r$ be the number of votes that $c_\kappa$ obtains in~$\E$ from voters in~$V_{\leq \kappa}$ and from voters in~$V_{\geq \kappa}$, respectively. The voters counted in both of these sets are exactly those in~$V^{P_\kappa}$, who all vote for~$c_\kappa$ in~$\E$. Hence, the score obtained by~$c_\kappa$ is ${s_l+s_r-|V^{P_\kappa}|=s^\star}$. Moreover, as $c_\kappa$ is a winner, no candidate may obtain a score higher than~$s$, and thus $(c_1,\dots,c_\kappa)$ and $(c_\kappa,\dots,c_k)$ is a left- and a right-PP-feasible partial nomination \scheme\  for~$c_\kappa$, respectively, in which $c_\kappa$ obtains $s_l$ votes from~$V_{\leq \kappa}$ and $s_r$ votes from~$V_{\geq \kappa}$, proving~(b).

     Assume now that (b) holds. Then there exists a left- and a right-PP-feasible partial nomination \scheme\ $(c_1,\dots,c_\kappa)$ and $(c_\kappa,\dots,c_k)$, respectively, for~$c_\kappa$, with $c_\kappa$ obtaining $s_l$ votes from~$V_{\leq \kappa}$ in the former and $s_r$ votes from~$V_{\geq \kappa}$ in the latter. Thus, in the nomination \scheme\ $(c_1,\dots,c_\kappa,\dots,c_k)$, every nominee obtains at most~$s^\star$ votes, and $c_\kappa$ obtains exactly $s_l+s_r-|V^{P_i}|=s^\star$ votes, and hence becomes a winner.
 \end{proof}

\paragraph{Computing viable scores.}
Let us now provide a dynamic programming method for computing all left-PP-viable scores for each candidate in~$P_1 \cup \dots \cup P_\kappa$; computing right-PP-viable scores for all candidates in~$P_\kappa \cup \dots \cup P_k$ can be done in a symmetric manner. Henceforth, we may write PP-viable instead of left-PP-viable.

First, it is clear that for each $c_1 \in P_1$, 
the only PP-viable score is $|V^{P_1}|$, 
because all voters in~$V_{\leq 1}=V^{P_1}$ vote for the nominee of~$P_1$ irrespective of the nomination \scheme.

For candidates belonging to some party~$P_i$ with $1<i\leq \kappa$, we use Lemma~\ref{lem:viable-score-recursion} for computing the PP-viable scores. 
\begin{lemma}
\label{lem:viable-score-recursion}
Let $1 <i \leq \kappa$, and consider some candidate~$c_i \in P_i$. A non-negative integer~$s_i$ is a PP-viable score for~$c_i$ if and only if 
there exist a candidate~$c_{i-1} \in P_{i-1}$ and an integer~$s_{i-1}$ between $0$ and~$s^\star$ such that 

\smallskip
\begin{minipage}{0.96\columnwidth}
\begin{itemize}
\item[(i)] $s_{i-1}$ is PP-viable for~$c_{i-1}$,
\item[(ii)] $c_i$ obtains exactly~$s'_i=s_i-|V^{P_i}|$ votes from $V^{(P_{i-1},P_i)}$ in any nomination \scheme\ where $c_{i-1}$ and $c_i$ are nominated, and
\item[(iii)] $s_{i-1}+|V^{(P_{i-1},P_i)}|-s'_i \leq s^\star$. 
\end{itemize}
\end{minipage}
\end{lemma}

\begin{proof}
    Given a left-PP-feasible partial nomination \scheme\ $(c_1,\dots,c_{i-1},c_i)$ for~$i$, observe that the votes obtained by~$c_i$ from voters in~$V_{\leq i}$ can be partitioned into those coming from voters in $V^{P_i}$ and those coming from voters in~$V^{(P_{i-1},P_i)}$. Hence, if $c_i$ obtains altogether $s_i$ votes from~$V_{\leq i}$, then it must obtain exactly $s'_i=s_i -|V^{P_i}|$ votes from voters in~$V^{(P_{i-1},P_i)}$, proving~(ii). Similarly, 
    the votes obtained by~$c_{i-1}$ from voters in~$V_{\leq i}$ can be partitioned into those coming from voters in~$V_{\leq i-1}$ and those coming from voters in~$V^{(P_{i-1},P_i)}$. 
    Let us denote the number of votes in the former category as $s_{i-1}$; the number of votes in the latter category is exactly $|V^{(P_{i-1},P_i)}|-s'_i$.
    Since $c_{i-1}$ cannot win more than~$s^\star$ votes, we obtain that (iii) holds. Finally, it is also clear from the definition of~$s_{i-1}$ and from the left-feasibility of $(c_1,\dots,c_{i-1},c_i)$ that $s_i$ is a PP-viable score for~$c_{i-1}$, with the partial nomination \scheme\ $(c_1,\dots,c_{i-1})$ as its witness.
    
    For the other direction, let $(c_1,\dots,c_{i-1})$ be a left-PP-feasible partial nomination \scheme\
    for~$i-1$ that witnesses (i), i.e., the PP-viability of~$s_{i-1}$ for $c_{i-1}$. Then $(c_1,\dots,c_{i-1},c_i)$ is a left-PP-feasible partial nomination \scheme\ for~$i$, because due to~(ii), $c_{i-1}$ obtains $|V^{(P_{i-1},P_i)}|-s'_i$ votes from voters in~$V^{(P_{i-1},P_i)}$, and thus obtains at most~$s^\star$ votes in total from voters in~$V_{\leq i}$ due to~(iii). Finally, (ii) also implies that $c_i$ obtains exactly~$s'_i+|V^{P_i}|=s_i$ votes from~$V_{\leq i}$ in every nomination \scheme\ containing candidates $c_1,\dots,c_{i-1},c_i$, which shows that $s_i$ is a PP-viable score for~$c_i$.
\end{proof}

Lemma~\ref{lem:viable-score-recursion} yields a way to compute all PP-viable scores between $0$ and~$s^\star$ for each candidate in~$P_1 \cup \dots \cup P_\kappa$ using dynamic programming: we need to compute and store a binary variable for each $s \in \{0,\dots,s^\star\}$ and each $c_i \in P_i$, $i \in [\kappa]$, describing whether $s$ is PP-viable for~$c_i$ or not; we already know these values for $i=1$, and we can compute them for increasing values of~$i$ using Lemma~\ref{lem:viable-score-recursion}. 
Having computed all left-PP-viable scores and all right-PP-viable scores (in a symmetric manner) for each candidate of~$P_\kappa$, we can use Lemma~\ref{lem:viable-score-vs-possible-winner} to find a nominee for~$P_\kappa$ that can win in some reduced election, or conclude correctly that no such candidate exists.

It is straightforward to see that the running time of the presented algorithm for the \pospres\ problem is polynomial.

\subsection{Solving the \necpres\ problem}
    By slightly modifying our algorithm for \textsc{Possible President}, 
    we can also decide for each pair~$(P^+,P^-)$ of parties whether there is a nomination \scheme\ that results in $P^+$ becoming a winner but $P^-$ not: during the computation of PP-viable scores for each candidate, we simply need to exclude~$s^\star$, the winning score, from the set of PP-viable scores for all candidates of~$P^-$. Therefore, to \necpres\ in polynomial time under party-aligned single-peaked preferences, we simply need to decide for each party~$P' \in \P \setminus \{P\}$ whether $P'$ can become a winner in some reduced election where $P$ is not a winner. If so, then our instance for \textsc{Necessary Winner} is a no-instance; otherwise, if no party~$P'$ can become a winner while also preventing $P$ from winning, then we have a yes-instance. 



\end{appendices}
\fi 

\end{document}